\documentclass{article}

\usepackage{microtype}
\usepackage{graphicx}
\usepackage{subfigure}
\usepackage{booktabs} 
\usepackage{xr-hyper}

\usepackage{hyperref}

\usepackage[utf8]{inputenc}
\usepackage{amsfonts}
\usepackage{enumerate}
\usepackage{graphicx}
\usepackage{xspace}
\usepackage{stmaryrd}
\usepackage{ae}
\usepackage{bbm}
\usepackage{amsmath}
\usepackage{amsthm}

\usepackage{enumitem}

\usepackage{mathtools}

\newtheorem{lemma}{Lemma}
\newtheorem{theorem}{Theorem}

\newtheorem{remark}{Remark}

\allowdisplaybreaks

\makeatletter
\newcommand{\newreptheorem}[2]{\newtheorem*{rep@#1}{\rep@title}
\newenvironment{rep#1}[1]{\def\rep@title{#2 \ref*{##1}}\begin{rep@#1}}{\end{rep@#1}}
}
\makeatother

\newreptheorem{lemma}{Lemma}
\newreptheorem{theorem}{Theorem}
\newreptheorem{claim}{Claim}
\newreptheorem{proposition}{Proposition}

\DeclareMathOperator*{\argmin}{arg\,min}

\newcommand{\RR}[0]{\mathbb{R}}
\newcommand{\ip}[2]{\left\langle #1, #2 \right\rangle}

\def\l({\left(}
\def\r){\right)}
\def\bl({\Big(}
\def\br){\Big)}
\def\beq{\begin{equation}}
\def\eeq{\end{equation}}

\usepackage{tikz,pgfplots}
\def\x{{\mathbf x}}
\def\y{{\mathbf y}}
\def\g{{\mathbf g}}

\def\z{{\mathbf z}}

\def\RR{{\mathbb R}}

\def\w{{\mathbf w}}

\def\wt{{w_{t}}}
\def\wt1{{w_{t+1}}}
\def\at1{{a_{t+1}}}

\def\Dcal{{ \mathcal{D} }}

\def\nn{{ \nonumber }}

\usepackage[accepted]{icml2018}

\icmltitlerunning{Let's be Honest: An Optimal No-Regret Framework for Zero-Sum Games}

\begin{document}

\twocolumn[
\icmltitle{Let's be Honest: An Optimal No-Regret Framework for Zero-Sum Games}



\icmlsetsymbol{equal}{*}

\begin{icmlauthorlist}
\icmlauthor{Ehsan Asadi Kangarshahi}{equal,epfl}
\icmlauthor{Ya-Ping Hsieh}{equal,epfl}
\icmlauthor{Mehmet Fatih Sahin}{epfl}
\icmlauthor{Volkan Cevher}{epfl}
\end{icmlauthorlist}

\icmlaffiliation{epfl}{LIONS, EPFL, Switzerland}

\icmlcorrespondingauthor{Ya-Ping Hsieh}{ya-ping.hsieh@epfl.ch}
\icmlcorrespondingauthor{Volkan Cevher}{volkan.cevher@epfl.ch}

\icmlkeywords{Zero-Sum Games, No-Regret Algorithms}

\vskip 0.3in
]



\printAffiliationsAndNotice{\icmlEqualContribution} 

\begin{abstract}
We revisit the problem of solving two-player zero-sum games in the decentralized setting. We propose a simple algorithmic framework that simultaneously achieves the best rates for honest regret as well as adversarial regret, and in addition resolves the open problem of removing the logarithmic terms in convergence to the value of the game. We achieve this goal in three steps. First, we provide a novel analysis of the optimistic mirror descent (OMD), showing that it can be modified to guarantee fast convergence for both honest regret and value of the game, when the players are playing collaboratively. Second, we propose a new algorithm, dubbed as robust optimistic mirror descent (ROMD), which attains optimal adversarial regret without knowing the time horizon beforehand. Finally, we propose a simple signaling scheme, which enables us to bridge OMD and ROMD to achieve the best of both worlds. Numerical examples are presented to support our theoretical claims and show that our non-adaptive ROMD algorithm can be competitive to OMD with adaptive step-size selection. 
%
\end{abstract}

\section{Introduction}
The simple zero-sum games have been studied extensively, often from the standpoint of analyzing the convergence to the Nash equilibrium. At the equilibrium, the players employ a min-max pair of strategies where no player can improve their pay-off by a unilateral deviation \citep{von1928theory}. 

In this setting, one can expect that the players arrive at the equilibrium via decentralized, no-regret learning algorithms, which hold even in the presence of potential adversarial behavior, and which also better model selfish play. The resulting dynamics is of great interest in optimization and behavioral economics \citep{myerson1999nash}, especially under communication constraints. 

When the behavior of each player is explained by a no-regret algorithm, it is possible to significantly improve convergence rates beyond the so-called black-box, adversarial dynamics. This observation was first made by \citep{daskalakis2011near}, which tailored a decentralized version of Nesterov's primal-dual method based on the excessive gap condition. 

Intriguingly, \citep{daskalakis2011near} left it as an open question on the existence of a simple algorithm that converges at optimal rates for both regret and the value of the game in an uncoupled manner, both against honest (i.e., cooperative) and dishonest (i.e., arbitrarily adversarial) behavior. 

The challenge was partially settled by the modified optimistic mirror descent (OMD) framework in \citep{rakhlin2013optimization}. While the framework of \citep{daskalakis2011near} is considered unnatural and involves additional logarithmic factors, similar arguments apply to \citet{rakhlin2013optimization}'s framework: The modified OMD needs to know the game horizon a priori to determine the step-sizes. Their analysis also results in non-optimal regret and logarithmic factors in convergence to the value of the game. 

Besides the aforementioned drawbacks, neither approaches can accommodate natural switches between honest and dishonest behavior. 

In this work, we propose a simple algorithmic framework that closes the gap between upper and lower bounds for adversarial regret as well as convergence to the value of the game, while maintaining the best known rate for honest regret, thereby resolving the open problem posed by \citep{daskalakis2011near}.


We achieve the desiderata as follows: First, we provide a novel analysis of OMD and show that it can obtain fast convergence for both honest regret and value of the game, when both players are honest. Second, we introduce robust optimistic mirror descent (ROMD), which attains optimal adversarial regret without knowing the time horizon. Finally, we propose a simple signaling scheme, which enables us to bridge OMD and ROMD to achieve the best of both worlds, and seamlessly handle honest and dishonest behavior.

\begin{table*}[!h]
\centering
\begin{tabular}{|l||*{5}{c|}}\hline
 &\makebox[6em]{\small\textbf{Honest $R_T$}}&\makebox[6em]{\small\textbf{Adversarial $R_T$}}&\makebox[6em]{\small\textbf{Game Value}} &\makebox[6em]{\small\textbf{Oracle}}&\makebox[6em]{\textbf{\small Algorithm}}  \\\hline \hline
\small\citet{daskalakis2011near} & {\small$O(\log T)$} &  \small $O(\sqrt{T})$ &  \small$O(T^{-1}{\log^{\frac{3}{2}} T})$  &  \small$|A|_{\max}$ &Complicated\\\hline
\small\citet{rakhlin2013optimization} &  $?$  &   \small$O(\sqrt{T}\log T)$ & \small$O\l(T^{-1}{\log T}\r)$ & \small$T, |A|_{\max}$ &Simple \\\hline 
\small This paper &  \small$O(\log T)$   &  \small$O(\sqrt{T})$ & \small$O\l(T^{-1}\r)$ & \small$|A|_{\max}$ &Simple \\\hline  
\end{tabular}\caption{\label{tab:convergence_comparison} A convergence rate comparison in the context of assumptions. }
\end{table*}

\subsection{Related Work}
\noindent\textbf{Algorithms for Decentralized Games:} To our knowledge, the only two explicit algorithms capable of solving zero-sum games in the decentralized setting are given by \citep{daskalakis2011near} and \citep{rakhlin2013optimization}, respectively. A comparison of their convergence rates versus ours is presented in Table 1.

The algorithm of \citep{daskalakis2011near} is a decentralized primal-dual method based on Nesterov's excessive gap technique \citep{nesterov2005excessive}. Its convergence guarantees are only slightly worse than ours (\textit{cf.,} Table \ref{tab:convergence_comparison}). However, due to the presence of complicated and unnatural scheduling steps, the authors in \citep{daskalakis2011near} themselves were not convinced by the practicality of their algorithm and stated the result as merely an ``existence proof.''

Later on, \citet{rakhlin2013optimization} proposed an algorithm based on the Optimistic Mirror Descent (OMD), initially introduced in a special case by \citep{chiang2012online} and also studied in detail by \citep{rakhlin2013online}. While the algorithm is simple, it features  several drawbacks. Foremost, it requires the time horizon beforehand, which is unsatisfactory. Second, when both players are playing collaboratively, their regret is sub-optimal. Third, its adversarial regret and convergence to the game value has extra $\log T$ factors, which require additional cautions to remove. Finally, the algorithm uses \emph{adaptive} step-sizes, requiring additional work per-iteration.

\noindent\textbf{Meta-Algorithms:} There exist some work on ``meta-algorithms'' for games \citep{syrgkanis2015fast, foster2016learning}, which can turn certain learning algorithms into solving zero-sum games. For instance, leveraging the framework in \citep{syrgkanis2015fast}, one can modify OMD to achieve $O(T^{\frac{1}{4}})$ for honest regret + $\tilde{O}\l( \sqrt{T}\r)$ for adversarial regret. Our algorithm uniformly outperforms these rates.


\section{Preliminaries and Notation}
Let $\psi$ be a mirror map over the convex domain $\Dcal$, and let $D(\cdot, \cdot)$ be the Bregman divergence associated with $\psi$. We assume the knowledge of the three-point identity for Bregman divergence in the sequel:
\beq \nn
D(\x,\y) + D(\y,\z) = D(\x,\z) + \ip{\x-\y}{  \nabla \psi (\z) - \nabla \psi (\y) }.
\eeq

We use the notation $\z = MD_\eta(\x,\g)$ to denote:
	\begin{align*}	
	\z= \nabla \psi^\star \Big( \nabla\psi(\x) - \eta \g \Big)
\end{align*}where $\psi^\star$ is the Fenchel dual of $\psi$. 

Let $\psi$ be 1-strongly convex with respect to the norm $\| \cdot \|$. We define 
\beq \nn
D^2 \coloneqq \max \left\{  \sup_{\x, \x' \in \Dcal} \frac{1}{2}\|\x-\x'\|^2, \sup_{\x \in \Dcal}D(\x, \x_c) \right\}
\eeq
where $\x_c \coloneqq \argmin_{\x\in \Dcal} \psi(\x)$ is the prox center. Hence $D$ controls both the diameter (in $\|\cdot\|$) and the Bregman divergence to the prox center.

We frequently use the fact that
\beq \nn
\ip{\x}{A \y} \leq | A |_{\max} \quad  \forall \x \in \Delta_m, \ \y \in \Delta_n
\eeq
where $| A |_{\max}$ is the maximum entry of $A$  in absolute value, and $\Delta_m \coloneqq \{ \x \in \RR^m \ | \ \sum_{i=1}^m x_i = 1, x_i \geq 0  \}$ is the standard simplex. On a simplex, we will only consider the entropic mirror map:
\beq \nn
\psi(\x) = \sum_{i=1}^k x_i \log x_i, \quad k = m \text{ or } n
\eeq
which is well-known to be 1-strongly convex in $\| \cdot \|_1$.

We use $\frac{1}{m}1_m$ to denote the uniform distribution on $\Delta_m$.


\section{Problem Formulation and Main Result}
An (offline) two-player zero-sum game with payoff matrix $A$ refers to the solving the minimax problem:
\beq \label{eq:two_player_zero_sum_game}
V \coloneqq \min_{\y \in \Delta_n} \max_{\x \in \Delta_m} \ip{\x}{A\y}.
\eeq
The quantity $ V$ in \eqref{eq:two_player_zero_sum_game} is called the \textbf{value} of the game, or the Nash Equilibrium Value. Any pair $(\bar{\x}, \bar{\y})$ attaining the game value is called an equilibrium strategy.

In the decentralized setting (aka., the ``strongly uncoupled'' setting),  the payoff matrix and the number of opponent's strategies are unknown to both players, and their goal is to learn a pair of equilibrium strategy through repeated game plays. Moreover, each player aims to suffer a low individual regret, even in the presence of an adversary or  a corrupted channel that distorts the feedback.

Specifically, at each round $t$, the players take actions $\x_t$ and $\y_t$, and then receive the loss vectors $-A\y_t$ (for $\x$-player) and $A^\top \x_t$ (for $\y$-player). In the honest setting, we assume that the two players take actions according to a prescribed algorithm, and we say the setting is adversarial if only one player (the $\x$-player in this paper) adheres to the prescribed algorithm and the other player arbitrary.

As in previous work, we assume that an upper bound $|A|_{\textup{max}}$ on the maximum absolute entry of $A$ is available to both players. The goal is to achieve
\begin{align*}
\left| V-   \ip{\x_T }{A\y_T} \right| &\leq r_1(T), \\
R_T \coloneqq \max_{\x \in \Delta_m} \sum_{t=1}^T \ip{\x_t - \x}{ -A\y_t} & \leq r_2(T)
\end{align*}for fast-decaying $r_1$ and sublinear $r_2$ in $T$. The first requirement is to approximate the game value in \eqref{eq:two_player_zero_sum_game}, and the second one asks to minimize the regret $R_T$.

Our main result can be stated as follows:
\begin{theorem}[Main result, informal]\label{thm:informal}
For \eqref{eq:two_player_zero_sum_game}, there is a simple decentralized algorithm with non-adaptive step-size such that
\begin{align*}
r_1(T) = O\l(\frac{1}{T}\r), \quad \quad r_2(T) = O\l(\log T\r),
\end{align*}if the opponent is honest (i.e., playing collaboratively to solve the game). Moreover, against any adversary, we have
\beq
r_2(T) = O\l(\sqrt{T}\r). \nn
\eeq
\end{theorem}

Except for the $O\l({\log T}\r)$ honest regret, these rates are known to be optimal \citep{cesa2006prediction,daskalakis2015near}. We are also the first to remove $\log T$ factors in convergence to the value of the game, an open question posed by the very first work in learning decentralized games \citep{daskalakis2011near}.


\section{A family of optimistic mirror descents: Classical, Robust, and Let's be honest}
We first illustrate the high-level ideas to prove \textbf{Theorem \ref{thm:informal}} in Section \ref{subsec:ideas}. A novel analysis for OMD in the honest setting is given in Section \ref{subsec:omd}, and we propose a new algorithm for the adversarial setting in Section \ref{subsec:romd}. Finally, the full algorithm is presented in Section \ref{subsec:full}, along with the rigorous version of the main result (\textit{cf.,} \textbf{Theorem \ref{thm:ADMD}}).

\subsection{High-Level Ideas}\label{subsec:ideas}
Our algorithms are inspired by the iterates of the form:
\beq \label{eq:omd_equiv}
\left\{
\begin{array}{ll}
\x_{t+1} = MD_{\eta}(\x_t, -2A\y_t + A\y_{t-1})\\
\y_{t+1} = MD_{\eta}(\y_t, 2A^\top \x_t - A^\top\x_{t-1})
\end{array},
\right.
\eeq
which are equivalent to the OMD in \citep{rakhlin2013optimization} (see Appendix \ref{app:quiv}). It is known that directly applying \eqref{eq:omd_equiv} to \eqref{eq:two_player_zero_sum_game} yields $O\l(\frac{1}{T}\r)$ convergence in the game value, however without any guarantee on the regret.

To make OMD optimal for zero-sum games, we improve \eqref{eq:omd_equiv} on two fronts. First, in the honest setting, we make the following simple observation: Although the iterates $\x_{t}$ are not guaranteed to possess sublinear regret, the averaged iterates $\frac{1}{t}\sum_{i=1}^t \x_i$ do enjoy logarithmic regret, and hence, it suffices to play the averaged iterates in the honest setting. 

Second, in order to make OMD robust against any adversary, we utilize the ``mixing steps'' of \citep{rakhlin2013optimization} with an important improvement: Our step-sizes do not depend on the time horizon. This new feature is crucial in removing $\log T$ factors in both the convergence to game value and adversarial regret. In fact, our analysis is arguably simpler than \citep{rakhlin2013optimization}.

%
%

\subsection{Optimistic Mirror Descent}\label{subsec:omd}
\begin{algorithm} [h]
		\caption{Optimistic Mirror Descent: $\x$-Player}
		Set $\eta= \frac{1}{2|A|_{\max}}$\\
		Play $\z_1 = \z_2 = \z_3 = \frac{1}{m}1_m $\\
		For $t \geq 3$:
		\begin{algorithmic}[1]
			\STATE Compute \begin{align*} \x_{t+1} = MD_\eta &(\x_t,-2(t-2)A\w_t  \\&+3(t-3)A\w_{t-1} - (t-4)A\w_{t-2} ) 
			\end{align*}
			\STATE Play $\z_{t+1} = \frac{1}{t-1}\sum_{i = 3}^{t+1}\x_i$
			\STATE Observe $-A\w_{t+1}$
		\end{algorithmic}\label{alg:HDMDx}
	\end{algorithm}
	
	\begin{algorithm}[h]
	\caption{Optimistic Mirror Descent: $\y$-Player}
	Set $\eta= \frac{1}{2|A|_{\max}}$\\
	Play $\w_1 = \w_2 = \w_3 = \frac{1}{n}1_n$ \\
	For $t \geq 3$:
		\begin{algorithmic}[1]
			\STATE	Compute \begin{align*} \y_{t+1} = MD_\eta &(\y_t,2(t-2)A^\top\z_t \\ &-3(t-3)A^\top\z_{t-1} + (t-4)A^\top\z_{t-2}) 
			\end{align*}
			\STATE	Play $\w_{t+1} = \frac{1}{t-1}{\sum_{i = 3}^{t+1}\y_i}$
			\STATE Observe $A^\top \z_{t+1}$
		\end{algorithmic}\label{alg:HDMDy}
	\end{algorithm}
	

As alluded to in Section \ref{subsec:ideas}, we will play OMD with the averaged iterates. The algorithms are given explicitly in \textbf{Algorithm \ref{alg:HDMDx}} and \textbf{\ref{alg:HDMDy}}. 
	
\begin{remark}
Note that there is no need to play $\frac{1}{m}1_m$ and $\frac{1}{n}1_n$ three times in  \textbf{Algorithm \ref{alg:HDMDx}} and \textbf{\ref{alg:HDMDy}}. The players could just play once $\l(\frac{1}{m}1_m\r)^\top A \l(\frac{1}{n}1_n\r)$ and would have enough information to run OMD from $\x_4$ and $\y_4$. Our choices are motivated by the resulting ease of the notation. 
\end{remark}

We analyze our version of OMD below. The crux of our analysis is to first look at the regrets of auxiliary sequences $\x_t$ and $\y_t$, and we show that the \emph{sum} of the auxiliary regrets, not any individual of them, controls both the convergence to the value of the game and the honest regret for the averaged sequences $\z_t$ and $\w_t$.
\begin{theorem}\label{thm:HDMD}
Suppose two players of a zero-sum game have played $T$ rounds according to the OMD algorithm with $\eta = \frac{1}{2|A|_{\max}}$. Then
\begin{enumerate}
\item The $\x$-player suffers an $O\l({\log  T}\r)$ regret:
\begin{align} 
\max_{\z\in\Delta_m}\sum_{t=3}^T \ip{\z_t - \z}{ -A\w_t} &\leq    \log 2(T-2)|A|_{\max} \times \nn\\
&\hspace{25pt} \Big(20 + \log m +\log n \Big) \label{eq:HDMD_low_regret}\\
&= O\l( {\log T}\r) \nn
\end{align}
and similarly for the $\y$-player.
\item The strategies $(\z_T,\w_T)$ constitutes an $O\l(\frac{1}{T} \r)$-approximate equilibrium to the value of the game:
\begin{align} 
\left| V -   \ip{\z_T}{ A \w_T} \right| &\leq \frac{ \Big(20 + \log m +\log n \Big)|A|_{\max}} {T-2}\label{eq:HDMD_value_of_the_game}\\
& = O\l(  \frac{1}{T}\r). \nn
\end{align}
\end{enumerate}
\end{theorem}
\begin{proof}
See \textbf{Appendix \ref{app:HDMD_proof}}.
\end{proof}

\subsection{Robust Optimistic Mirror Descent}\label{subsec:romd}
In this section, we introduce  \textbf{robust optimistic mirror descent} (ROMD), which is a novel algorithm even for online convex optimization. 

Let $\psi$ be 1-strongly convex with respect to $\| \cdot \|$, and suppose we are minimizing the regret against an arbitrary sequence of convex functions $f_1, f_2, \ldots $ in a constraint set $\mathcal{D}$. Assume that each function is $G$-Lipschitz in $\| \cdot \|$. Assume also that no Bregman projection is needed (i.e., $MD_{\eta}(\x,\g) \in \Dcal$ for any $\x$ and $\g$); this is, for instance, the case for the entropic mirror map. 

We state ROMD in the general form in \textbf{Algorithm \ref{alg:RDMD}}.

\begin{theorem}[$O(\sqrt{T})$-Adversarial Regret]\label{thm:RDMD}
Suppose that $\| \nabla f_t \|_* \leq G$ for all $t$. Then playing $T$ rounds of \textbf{Algorithm \ref{alg:RDMD}} with $\eta_t = \frac{1}{G\sqrt{t}}$ against an arbitrary sequence of convex functions has the following guarantee on the regret:
\begin{align}
\max_{\x\in\Delta_m}\sum_{t=1}^T \ip{\x_t - \x}{ \nabla f_t (\x_t)} &\leq   G\sqrt{T}\l( 18+2D^2\r)  \nn\\
&\hspace{40pt}+ GD\l(3\sqrt{2} + 4D \r) \nn \\
& = O\l( {\sqrt{T} }\r).  \nn
\end{align}
\end{theorem}

\begin{proof}
See \textbf{Appendix~\ref{app:RDMD_proof}}.
\end{proof}

\begin{algorithm}[h]
	\caption{Robust Optimistic Mirror Descent}
	\begin{algorithmic}[1]
	\STATE Initialize $\x_1 = \x_c$, $\nabla f_{0}=0$, $\eta_t = \frac{1}{G\sqrt{t}}$\\
	\FOR{$t = 1, 2, ...$, }
		\STATE $\tilde{\x}_t = (\frac{t-1}{t}) \x_t + \frac{1}{t} \x_c$
		\STATE Set $\tilde{\nabla}_t = 2\nabla f_{t}(\x_t) - \nabla f_{t-1} (\x_{t-1})$, \\ play $\x_{t+1} = MD_{\eta_t}(\tilde{\x}_t, \tilde{\nabla}_t)$
		\STATE Observe $f_{t+1}$
		\ENDFOR
	\end{algorithmic}\label{alg:RDMD}
\end{algorithm}

When specialized to zero-sum games, it suffices to take $\x_c = \frac{1}{m}1_m$, $G = |A|_{\max}$, $D = \log m$, and $\psi$ being the entropic mirror map.

\begin{remark}
Our analysis of ROMD crucially relies on the assumption that no Bregman projection is needed. We have not been able to generalize our analysis to the case with Bregman projections.
\end{remark}


\def\p{{  \mathbf{p} }}
\def\xsig{{ \check{\x} }}
\def\ysig{{ \check{\y} }}
\subsection{Let's be honest: The full framework}\label{subsec:full}

%
We now present our approach for solving \eqref{eq:two_player_zero_sum_game}.

To ease the notation, define 
\beq
\z^*_t  \coloneqq \arg\min_{\x\in\Delta_m} \ip{\x}{-A\w_t} \nn
\eeq 
and 
\beq
\w^*_t  = \arg\min_{\y\in\Delta_n} \ip{\z_t}{A\y}. \nn
\eeq 
Let constants $C_1, C_2$, and $C_3$ be such that (see \textbf{Thereom \ref{thm:HDMD}}, \textbf{Theorem \ref{thm:RDMD}}, and \eqref{eq:HDMD_proof_hold5})
\begin{align}
&\ip{\z_t - \z_t^*}{ -A \w_t} \leq \frac{C_1}{t}  ,\ \quad \z_t, \w_t \text{ from OMD},  \label{eq:ADMD_hold1}\\
&\ip{\w_t - \w_t^*}{ A^\top \z_t} \leq \frac{C_1}{t}  , \quad \z_t, \w_t \text{ from OMD}, \label{eq:ADMD_hold2}\\
&\sum_{t=1}^T \ip{\z_t - \z^*}{ -A \y_t} \leq C_2 \sqrt{T}, \quad \z_t \text{ from ROMD and } \nn\\ 
&\hspace{140pt}\y_t \text{ arbitrary}, \label{eq:ADMD_hold4} \\
&| V - \z_T A \w_T| \leq \frac{C_3}{T}, \quad \z_T, \w_T \text{ from OMD}. \label{eq:ADMD_hold3}
\end{align}

From a high-level, our approach exploits the following simple observation: Suppose that we know $C_1$ above. If the instantaneous regret bound \eqref{eq:ADMD_hold1} and \eqref{eq:ADMD_hold2} hold true for all $t$, then we would trivially have the desired convergence. 

In contrast, if at any round the bound \eqref{eq:ADMD_hold1} is violated for the $\x$-player, then it must be due to an adversarial play, and we can simply switch to ROMD to get $O(\sqrt{T})$ regret. However, since $C_1$ (\textit{cf.}, \eqref{eq:HDMD_proof_hold5}) involves $n$, the number of opponent's strategies, the $\x$-player cannot compute it exactly. The situation is similar for the $\y$-player. We hence need to come up with a way to estimate $C_1$ for both players.
\begin{algorithm}[!t]
\caption{Let's Be Honest Optimistic Mirror Descent: $\x$-Player}
\begin{algorithmic}[1]
\STATE Initialize $b=1, t=1, \w_0 = \frac{1}{n}1_n$ and $\z_0 = \frac{1}{m}1_m$
\STATE Play $t$-th round of OMD-$\x$, observe $-A\p_t$
\STATE \IF {$G_t^{\w} \coloneqq \ip{\w_{t-1}}{A^\top\z_{t-1}} - \ip{\p_t }{A^\top\z_{t-1}} > \frac{b}{t-1}$ }{\STATE \hspace{20pt} Play $b^4 -1$ rounds of ROMD 
\STATE \hspace{20pt} $t \leftarrow t+1$ 
\STATE \hspace{23pt}$b \leftarrow 2b$
\STATE \hspace{20pt} Go to line 2.}  
\ENDIF
\STATE $-A\w_t \leftarrow -A\p_t$ 
\STATE \IF {$G_t^{\z} \coloneqq \ip{\z_{t}}{-A\w_{t}} - \ip{\z^*_t }{-A \w_{t}} > \frac{b}{t}$  }
{\STATE \hspace{20pt} Play $\xsig_{t+1} \coloneqq \z_t^*$
\STATE \hspace{20pt} Play $b^4 -1$ rounds of ROMD 
\STATE \hspace{20pt} $t \leftarrow t+2$ 
\STATE \hspace{23pt}$b \leftarrow 2b$
\STATE \hspace{20pt} Go to line 2.} 
\ENDIF
\STATE $t \leftarrow t+1$
\STATE Go to line 2.
\end{algorithmic}\label{alg:ADMD-x}
\end{algorithm}

It is important to note that one can not na\"ively estimate $C_1$ by binary search separately on both players. The reason, and the major difficultly to the above approach, is as follows: Since in general $\ip{\z_t - \z_t^*}{ -A \w_t} \neq \ip{\w_t - \w_t^*}{ A^\top \z_t}$, it could be the case that, at the same round, the $\x$-player detects a bad instantaneous regret and switch to ROMD, while the $\y$-player remains in OMD, even though two players are both honest. However, our entire analysis of OMD would breakdown if the OMD is not played cohesively. 

Furthermore, recall that we also want robustness against any adversary. Therefore, a bad instantaneous regret indicates the possibility of receiving an adversarial play, and we need to switch to ROMD whenever this occurs. 

To resolve such issues, we devise a simple \textbf{signaling} scheme ($\check{\x}_t$ and $\check{\y}_t$ below), which synchronizes both players' $C_1$ estimate and also the OMD plays while guaranteeing robustness. 

In words, our signaling scheme is a ``Let's be honest'' message to the opponent: ``I am having a bad instantaneous regret. Please update your $C_1$ with me, and please pretend that I am adversarial for a small number of rounds, so that we can play honest OMD cohesively.'' It turns out that doing these extra signaling rounds do not hurt the convergence rates in OMD and ROMD at all.

Our full algorithm, termed \textbf{Let's Be Honest} (LbH) \textbf{Optimistic Mirror Descent}, is presented in \textbf{Algorithm \ref{alg:ADMD-x}} and \textbf{\ref{alg:ADMD-y}}.
{\remark In \textbf{Algorithm \ref{alg:ADMD-x}} and \textbf{\ref{alg:ADMD-y}}, the role of $b$ is to estimate the constant $C_1$ in \eqref{eq:ADMD_hold1}. Since our analysis requires $b$ to be the same for both players throughout the algorithm run, a simple way is to assume that, say, $m=n = 5$, compute the corresponding $\tilde{C}_1$, and set the initial $b\leftarrow\tilde{C}_1$. Doing so indeed improves upon constants in our convergence; we chose $b=1$ only for simplicity.}
\begin{remark}
There are some degree of freedom in \textbf{Algorithm \ref{alg:ADMD-x}} and \textbf{\ref{alg:ADMD-y}}. For instance, instead of doubling $b$ in Line 16, one can do $b \leftarrow (1+\epsilon)b$ for some $\epsilon >0$. In Line 5, one can also play $b^2-1$ rounds, rather than $b^4-1$. As will become apparent in \textbf{Theorem \ref{thm:ADMD}}, these variants only effect the constants but not the convergence rates. However, they do have impact on empirical performance; \textit{cf.}, Section \ref{sec:experiments}.
\end{remark}

\begin{algorithm}[!t]
\caption{Let's Be Honest Optimistic Mirror Descent: $\y$-Player}
\begin{algorithmic}[1]
\STATE Initialize $b=1, t=1, \w_0 = \frac{1}{n}1_n$ and $\z_0 = \frac{1}{m}1_m$
\STATE Play $t$-th round of OMD-$\y$, observe $A^\top\mathbf{o}_t$
\STATE \IF{$G_t^{\z} \coloneqq \ip{\z_{t-1}}{-A\w_{t-1}} - \ip{\mathbf{o}_t }{-A\w_{t-1}} > \frac{b}{t-1}$ }{\STATE \hspace{20pt} Play $b^4 -1$ rounds of ROMD 
\STATE \hspace{20pt} $t \leftarrow t+1$ 
\STATE \hspace{23pt}$b \leftarrow 2b$
\STATE \hspace{20pt} Go to line 2.}  
\ENDIF
\STATE $A\z_t \leftarrow A^\top\mathbf{o}_t$ 
\STATE \IF {$G_t^{\w} \coloneqq \ip{\w_{t}}{A^\top\z_{t}} - \ip{\w^*_t }{A^\top \z_{t}} > \frac{b}{t}$  }
{\STATE \hspace{20pt} Play $\ysig_{t+1}  \coloneqq \w_t^*$
\STATE \hspace{20pt} Play $b^4 -1$ rounds of ROMD 
\STATE \hspace{20pt} $t \leftarrow t+2$ 
\STATE \hspace{23pt}$b \leftarrow 2b$
\STATE \hspace{20pt} Go to line 2.} 
\ENDIF
\STATE $t \leftarrow t+1$
\STATE Go to line 2.
\end{algorithmic}\label{alg:ADMD-y}
\end{algorithm}
The following key lemma ensures the two players to enter the ROMD plays coherently.
\begin{lemma}\label{lem:switching_together}
If the $\y$-player enters Line 12 of \textbf{Algorithm \ref{alg:ADMD-y}} at the $t$-th round, then the $\x$-player enters Line 4 of \textbf{Algorithm \ref{alg:ADMD-x}} at the $(t+1)$-th round. Conversely, if, at the $t$-th round, the $\y$-player does not enter Line 12 of \textbf{Algorithm \ref{alg:ADMD-y}}, then the $\x$-player does not enter Line 4 of \textbf{Algorithm \ref{alg:ADMD-x}} at the $(t+1)$-th round.

Exactly the same statements hold when the $\x$- and $\y$-player are reversed above.
\end{lemma}
\begin{proof}
If the $\y$-player enters Line 12 of \textbf{Algorithm \ref{alg:ADMD-y}} at the $t$-th round, then $\ysig_{t+1}$ is signalled at the $(t+1)$-th round, and it must be the case that $\ip{\w_{t} - \w_{t}^* }{A^\top\z_{t}} > \frac{b}{t}$ (\textit{cf.}, Line 12 of \textbf{Algorithm \ref{alg:ADMD-y}}). Therefore, at the $(t+1)$-th round, the $\x$-player would receive $-A\ysig_{t+1} = -A \w^*_t$ and compute 
\begin{align*}
G^{\w}_{t+1} &= \ip{\w_{t}}{A^\top\z_{t}} - \ip{\ysig_{t+1} }{A^\top\z_{t}} \\
&= \ip{\w_{t} - \w_{t}^* }{A^\top\z_{t}}> \frac{b}{t} 
\end{align*}which then enters the Line 4 of \textbf{Algorithm \ref{alg:ADMD-x}}.

Conversely, suppose that the $\y$-player does not enter Line 12 of \textbf{Algorithm \ref{alg:ADMD-y}} at the $t$-th round (or, equivalently, plays OMD at the $(t+1)$-th round). Then $\ip{\w_{t} - \w_{t}^* }{A^\top\z_{t}} \leq \frac{b}{t}$, implying that 
\begin{align*}
G^{\w}_{t+1} &= \ip{\w_{t} - \w_{t+1} }{A^\top\z_{t}}  \\
&\leq \ip{\w_{t} - \w_{t}^* }{A^\top\z_{t}} \leq \frac{b}{t}
\end{align*}hence preventing the $\x$-player from entering Line 4 of \textbf{Algorithm \ref{alg:ADMD-x}}.

Exactly the same computation holds when we reverse the role of $\x$- and $\y$-player.
\end{proof}
Given \textbf{Lemma \ref{lem:switching_together}}, we now know that the $\x$-player switches to ROMD \textbf{if and only if} the $\y$-player does. The rest of the proof then readily follows from \textbf{Theorems \ref{thm:HDMD}} and \textbf{\ref{thm:RDMD}}.

\begin{theorem}\label{thm:ADMD}
Suppose the $\x$-player plays according to \textbf{Algorithm \ref{alg:ADMD-x}} for $T$ rounds, and let $\textup{R}_T$ be the regret up to time $T$. Then
\begin{enumerate}
\item Let $T = T_1 + T_2 + T_3$ where $T_1$ is the number of OMD plays, $T_2$ is the number of ROMD plays, and $T_3$ is the number of signaling rounds (playing $\xsig_t$ or $\ysig_t$). Then there are constants $C$ and $C'$, depending only on $m,n$ and $|A|_{\max}$, such that
\beq \label{eq:regret_overall}
\frac{1}{T}\textup{R}_T \leq  \frac{C\log T_1 + C' \sqrt{T_2}}{T_1+ T_2}.
\eeq
In particular, if the opponent plays honestly, then $\textup{R}_T = O(\log T_1) = O(\log T). $ If the opponent is adversarial, we have $\textup{R}_T  = O(\sqrt{T_2}) = O(\sqrt{T}). $

\item Suppose that the honest $\y$-player plays \textbf{Algorithm \ref{alg:ADMD-y}}. Then the pair $(\z_T, \w_T)$ constitutes an $O\l( \frac{1}{T}\r)$-approximate equilibrium:
\beq  \label{eq:value_ADMD}
| V - \ip{\z_T}{A\w_T} | \leq \frac{C''}{T}
\eeq
for some constant $C''$.
%
\end{enumerate}
\end{theorem}
\begin{figure*}[!h]
    \centering
    \subfigure[{\label{fig:aq}}Value of the game.]{{\includegraphics[keepaspectratio=true,scale=0.65]{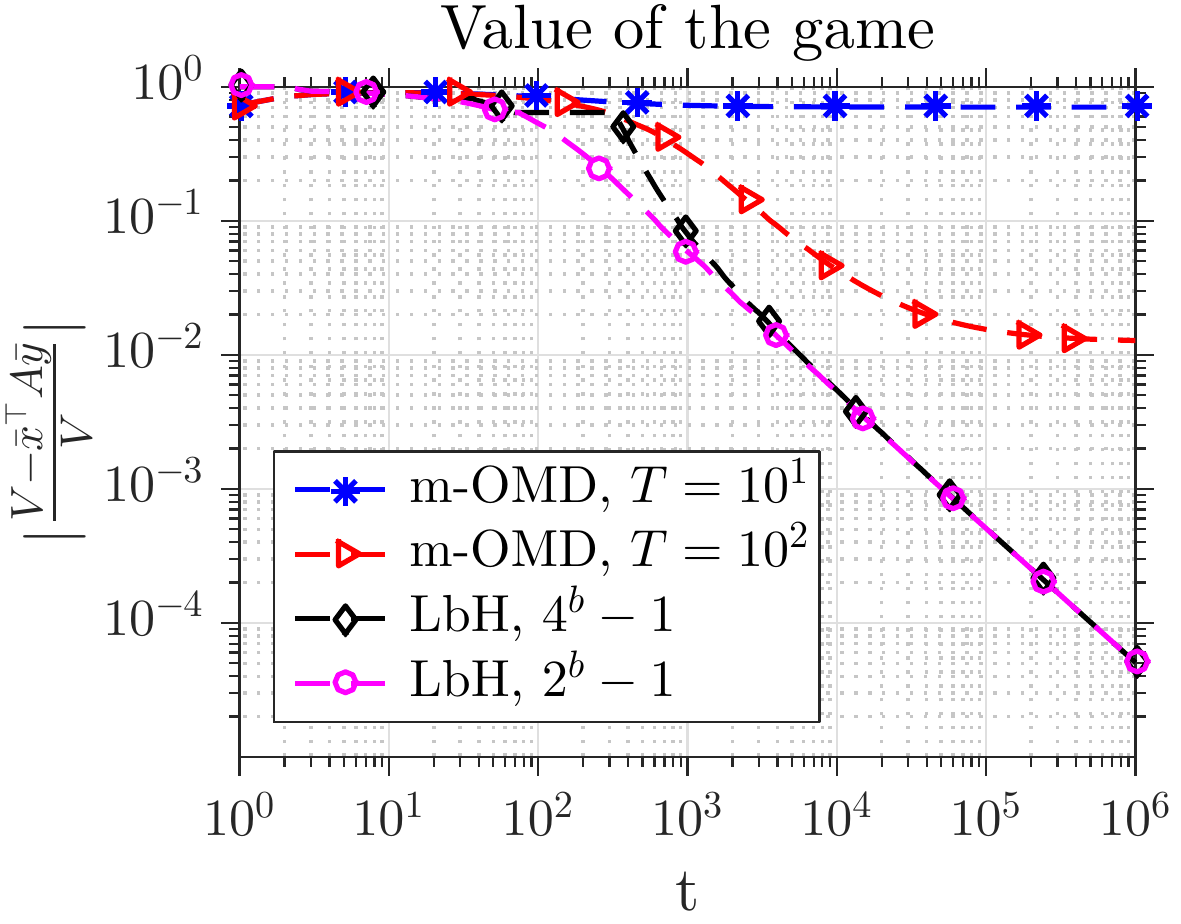} }}
    \subfigure[{\label{fig:bq}}Regret.]{{\includegraphics[keepaspectratio=true,scale=0.65]{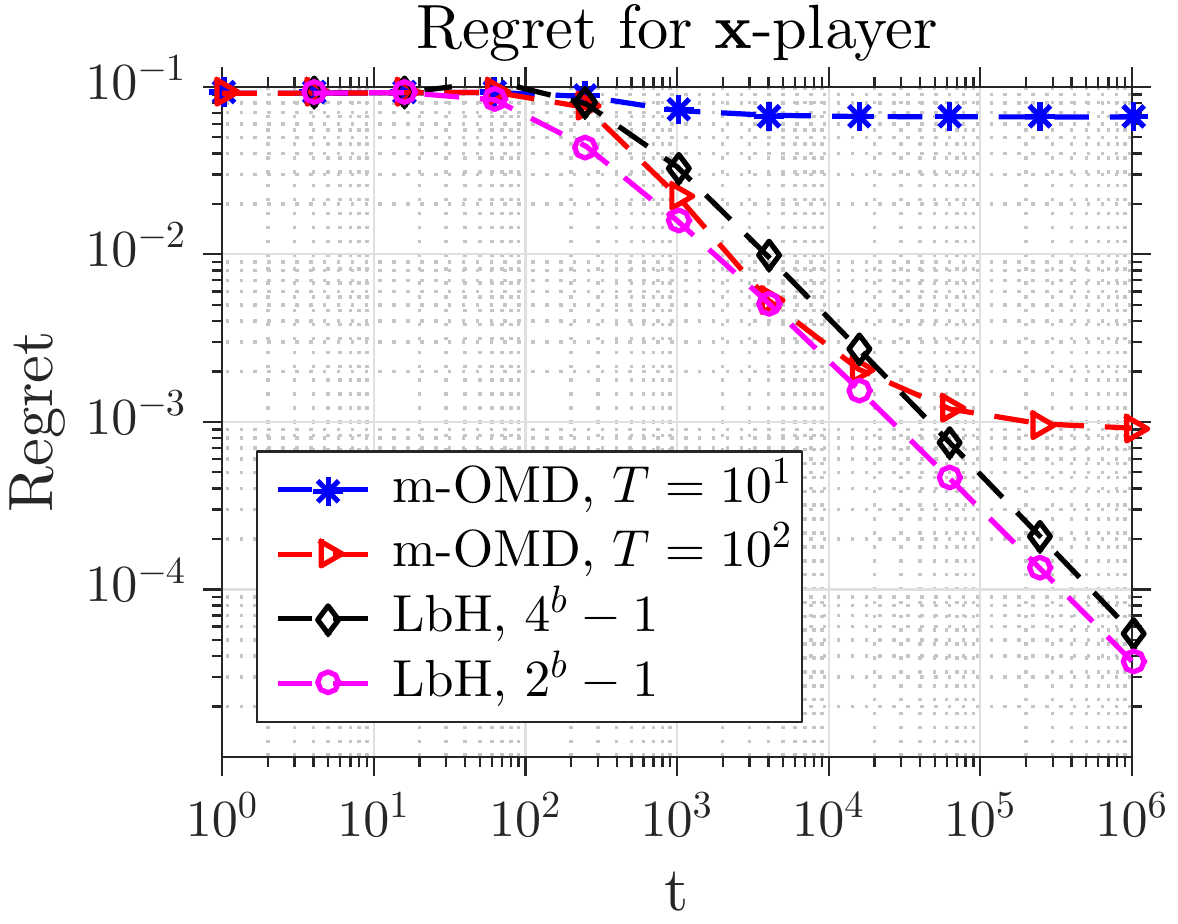} }}%
    \caption{Honest setting.}
    \label{fig:honest}
\end{figure*}
\begin{proof}
Suppose first that both players are honest.

We first prove the individual regret for the $\x$-player. We split the terms as follows:
\begin{align} \label{eq:ADMD_hold5}
\textup{R}_T &=    \textup{R}_{T_1}(\textup{playing OMD})  + \textup{R}_{T_2}(\textup{playing ROMD}) \nn\\
&\hspace{30pt} + \textup{R}_{T_3}(\textup{signaling}) . 
\end{align}
Recall \eqref{eq:ADMD_hold1}-\eqref{eq:ADMD_hold3}. We claim that 
\begin{enumerate}[label=(\alph*)]
\item $T_3 \leq \lceil \log C_1 \rceil. $
\item $T_2 \leq 16\cdot\frac{16^{T_3-1} -1}{15} \coloneqq C_1'$.
\end{enumerate}
Indeed, after $\lceil \log C_1 \rceil$-times signaling, we would have $b = 2^{T_3} > C_1$. Then \eqref{eq:ADMD_hold1} and \eqref{eq:ADMD_hold2} imply that we will never enter Line 12 again. On the other hand, we have
\beq \nn
T_2 \leq \sum_{r=1}^{T_3} 2^{4r} = \frac{16^{T_3-1} -1}{15}.
\eeq
Combining (a), (b) and using \eqref{eq:ADMD_hold1}, \eqref{eq:ADMD_hold4} in \eqref{eq:ADMD_hold5}, we conclude that
\begin{align*}
\textup{R}_T &\leq C_1\log T_1 + C_2\sqrt{T_2} + 2|A|_{\max}T_3 \\
& \leq C_1\log T_1  + C_2 \sqrt{C_1'} + 2|A|_{\max} \lceil \log C_1 \rceil \\
&= O(\log T_1) = O(\log T) 
\end{align*}which establishes \eqref{eq:regret_overall} in the honest case.

For convergence to the value of the game, we have, by \eqref{eq:ADMD_hold3},
\begin{align*}
| V - \ip{\z_T}{A\w_T} | \leq \frac{C_3}{T-T_2-T_3} \leq \frac{C_3}{T- C^*} 
\end{align*}where $C^* = \lceil \log C_1 \rceil + C_1'.$ The proof of \eqref{eq:value_ADMD} is completed by using the fact that $\frac{1}{T-C^*} \leq \frac{C^*}{T}$ when $T\geq \frac{C^{*2}}{C^*-1}$.

Finally, we show \eqref{eq:regret_overall} in the adversarial case.

Let $T_1, T_2$, and $T_3$ be as before, and we again split the regret into:
\begin{align}
\textup{R}_T &=    \textup{R}_{T_1}(\textup{playing OMD})  + \textup{R}_{T_2}(\textup{playing ROMD}) \nn\\
&\hspace{30pt} + \textup{R}_{T_3}(\textup{signaling}).  \nn
\end{align}
Notice that this time the inequalities \eqref{eq:ADMD_hold1} and \eqref{eq:ADMD_hold2} do not apply since the opponent no longer plays OMD collaboratively. However, by Line 12 of \textbf{Algorithm \ref{alg:ADMD-x}}, for every OMD play we must have 
\begin{align*}
\ip{\z_{t}}{-A\w_{t}} - \ip{\z^*_t }{-A \w_{t}} \leq \frac{b}{t} \leq \frac{2^{T_3}}{t}.
\end{align*}
Following the analysis as in the honest setting, we may further write
\beq \nn
\textup{R}_T \leq    2^{T_3}\log T_1 + C_2\sqrt{T_2} + 2|A|_{\max} T_3.
\eeq
It hence suffices to show that 
\beq \label{eq:ADMD_hold6}
2^{T_3} \log T_1 \leq C^{**}\sqrt{T_1 + T_2}.
\eeq
for some constant $C^{**}$.
To see \eqref{eq:ADMD_hold6}, recall that 
\beq \nn
T_2 = \frac{16(16^{T_3} -1)}{15} \geq 16^{T_3-1}.
\eeq
But then
\begin{align*}
\frac{2^{T_3} \log T_1 }{\sqrt{T_1 + T_2}} &\leq \frac{2^{T_3} \log T_1 }{\sqrt{2 \sqrt{T_1  T_2}}} \\
&\leq \frac{2^{T_3} \log T_1 }{2^{T_3-1} \cdot \sqrt{2}  \cdot \sqrt[4]{T_1} } \leq C^{**}.
\end{align*}for some universal constant $C^{**}$.

\end{proof}
\begin{remark}
As is evident from the proof, we have made no attempt to sharpening the constants, and hence our bounds can be numerically loose.
\end{remark}

\section{Experiments}\label{sec:experiments}
The purpose of this section is to provide numerical evidence to the following claims of our theory:
\begin{enumerate}
\item The LbH algorithm does not require knowing the time horizon beforehand, and our step-sizes are non-adaptive. Therefore, all quantities of interest, such as regrets or game value, should steadily decrease along the algorithm run.

\item The LbH algorithm automatically adjusts to honest and adversarial opponents.
\end{enumerate}

For comparison, we include the modified OMD (henceforth abbreviated as m-OMD) of \citep{rakhlin2013optimization} in our experiment, for different choices of time horizon. 


We generate the entries of $A$ uniformly at random in the interval $[-1, 1]$, and we set $m=200$ and $n=300$.


We consider two scenarios:
\begin{enumerate}
  \item \textit{Honest setting}: Both players adhere to the prescribed algorithms and try to reach the Nash equilibrium collaboratively.
  \item \textit{Adversarial setting}: The $\y$-player greedily maximizes the instantaneous regret of the $\x$-player.
\end{enumerate}

\begin{figure*}
    \centering
    \subfigure[{\label{fig:aq1}}Regret comparison.]{{\includegraphics[keepaspectratio=true,scale=0.65]{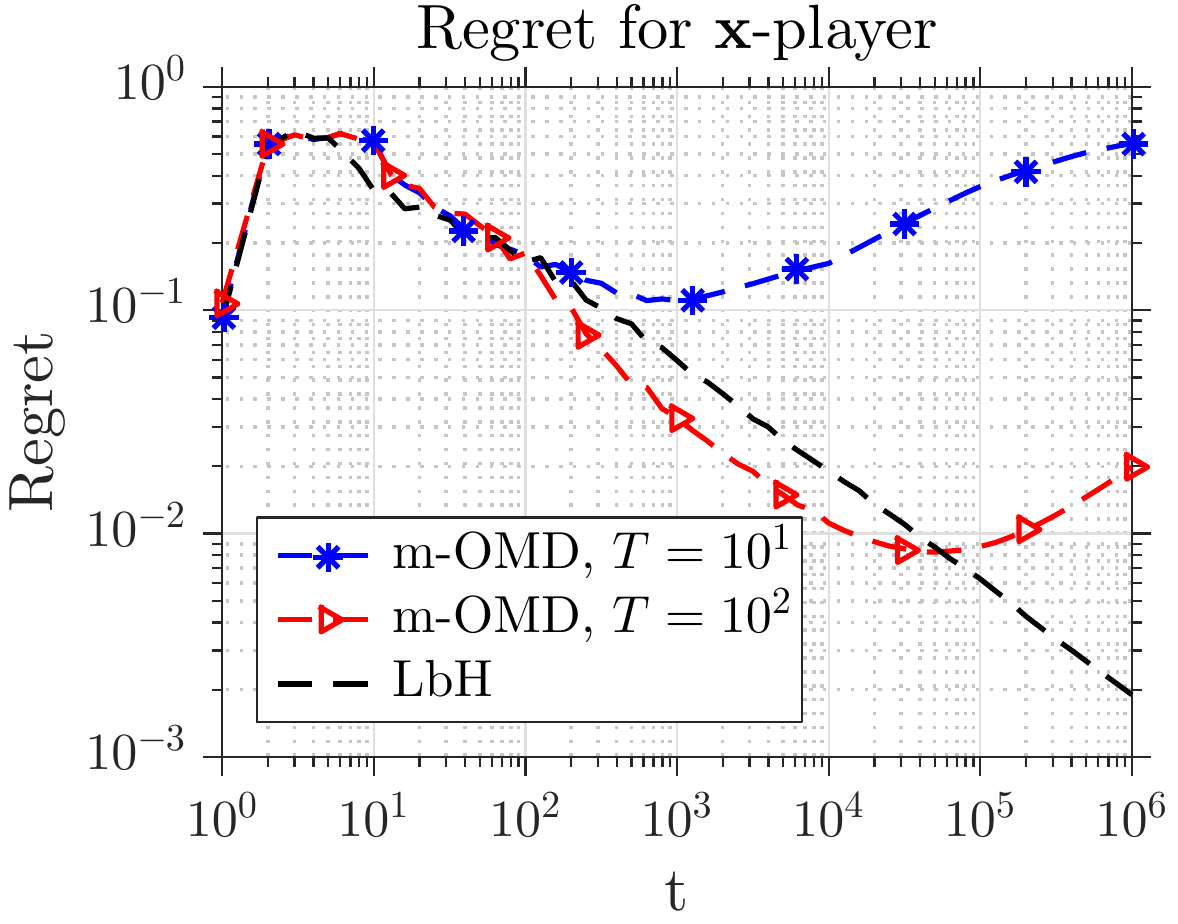} }}%
    \subfigure[{\label{fig:bq1}}Upper bound.]{{\includegraphics[keepaspectratio=true,scale=0.65]{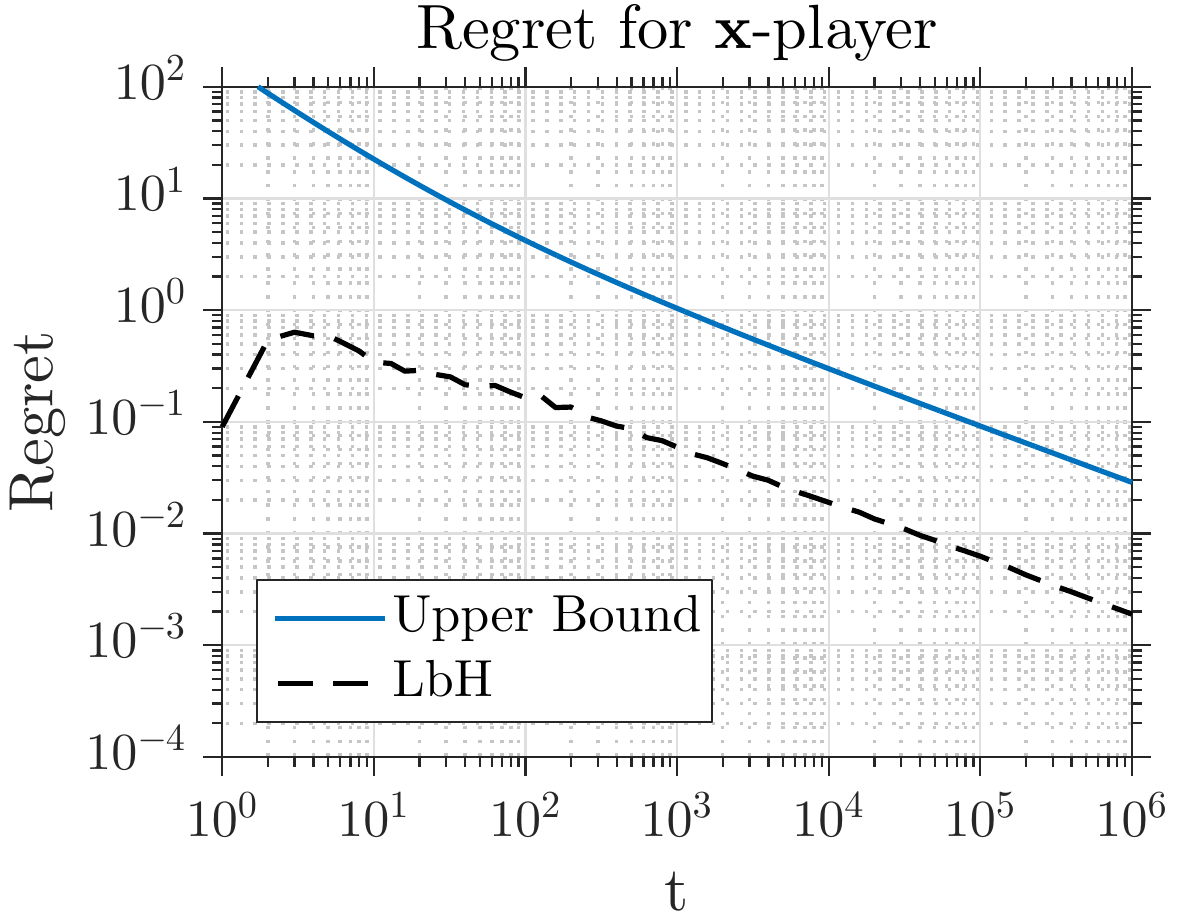} }}%
    \caption{Adversarial setting.}
    \label{fig:adv}
\end{figure*}
\subsection{Honest Setting}
The convergence for the honest setting is reported in \textbf{Figure \ref{fig:honest}}, for two different parameter choices of LbH and m-OMD. 

For both convergence to the game value and individual regret, after a short burn-in period (due to not knowing the $C_1$ in \eqref{eq:ADMD_hold1} and \eqref{eq:ADMD_hold2}), the LbH algorithm enters a steady $O\l(\frac{1}{T}\r)$-decreasing phase, as expected from our theory. 
~On the other hand, as the m-OMD chooses step-sizes according to the time horizon, it eventually saturates in both plots.

As noted by \citep{rakhlin2013optimization}, it is possible to prevent the saturation of m-OMD by employing the doubling trick or the techniques in \citep{auer2002adaptive}. However, doing so not only complicates the algorithm, but also introduces extra $\log T$ factors in the convergence of honest regret, since the doubling trick loses a $\log T$ factor for logarithmic regrets. Such rates are sub-optimal given our results.

\subsection{Adversarial Setting}
We report the regret comparison in \textbf{Figure \ref{fig:adv}}. 

In the adversarial setting, the LbH algorithm is essentially running the ROMD, and hence we see a straight $O({T}^{-\frac{1}{2}})$ decrease in the regret, as dictated by our upper bound in \textbf{Theorem \ref{thm:RDMD}}; see \textbf{Figure \ref{fig:adv}}-(b). The parameter choice does not effect the performance.

The m-OMD slightly outperforms LbH for a short period, but eventually blows up in regret. We remark that the short-term good empirical performance is due to the adaptive step-sizes of m-OMD, which require additional work per-iteration. Our LbH algorithm is non-adaptive, but is already competitive in terms of empirical performance.



\section{Conclusion and Future Work}
We studied the problem of zero-sum games in the decentralized setting, and we resolved an open problem of achieving optimal convergence to the game value while maintaining low regrets. Our techniques were based on several simple but novel observations in the game dynamics. Namely, we noticed that the averaged iterates of OMD enjoy logarithmic regret in the honest setting, we provided horizon-independent mixing steps for the OMD to achieve optimal adversarial regret, and we designed a singaling scheme to losslessly bridge OMD and ROMD. In essence, we showed that it is not necessary, as done in the work of \citep{rakhlin2013optimization}, to fix the time horizon beforehand and modify OMD accordingly. Our observations were instrumental in removing $\log T$ terms in all convergence rates.

Our framework suggests several research directions. First, instead of assuming that we observe the full loss vector, we may pose our problem in the \emph{bandit} setting, where only the payoff value of the current strategy is observed. Second, for practical purposes, it is interesting to see whether there exists an \emph{adaptive} step-size version of our algorithm. Finally, generalizing our framework to \emph{multiplayer} games is a challenging future work.

\section*{Acknowledgements}
This project has received funding from the European Research Council (ERC) under the European Union’s Horizon 2020 research and innovation programme (grant agreement n$^o$ 725594 - time-data), and was supported by the Swiss National Science Foundation (SNSF) under grant number 200021\_178865 / 1.

\bibliography{biblio}

\onecolumn
\begin{center}
\huge\textbf{Appendix}
\end{center}
\appendix

\numberwithin{equation}{section}
\section{Equivalence Formulations of Optimistic Mirror Descent}\label{app:quiv}
In this appendix, we show that the $\x_t$ iterates in \eqref{eq:omd_equiv} of the main text is equivalent to the following iterates given in \citep{chiang2012online, rakhlin2013online}:
\beq \label{eq:omd_equiv1}
\left\{
\begin{array}{ll}
\x_t &= MD_\eta\l(\tilde{\x}_t, -A\y_{t-1} \r)\\
\tilde{\x}_{t+1} &= MD_\eta \l(\tilde{\x}_t, -A\y_{t} \r)
\end{array}.
\right.
\eeq

By the optimality condition for \eqref{eq:omd_equiv1}, we have
\begin{align}
\nabla \psi(\x_t) &= \nabla \psi (\tilde{\x}_t) - \eta \l( -A\y_{t-1}\r),\label{eq:equiv_hold1}\\
\nabla \psi(\tilde{\x}_t) &= \nabla \psi (\tilde{\x}_{t-1}) - \eta \l( -A\y_{t-1}\r),\label{eq:equiv_hold2}\\
\nabla \psi(\tilde{\x}_{t-1}) &= \nabla \psi (\x_{t-1}) + \eta \l( -A\y_{t-2}\r). \label{eq:equiv_hold3}
\end{align}We hence get \eqref{eq:omd_equiv} by applying \eqref{eq:equiv_hold3} to \eqref{eq:equiv_hold2} and then \eqref{eq:equiv_hold2} to \eqref{eq:equiv_hold1}.
\section{Optimistic Mirror Descent}\label{app:HDMD_proof}
In this appendix, we prove \textbf{Theorem \ref{thm:HDMD}}, restated below for convenience.

\begin{theorem}
Suppose two players of a zero-sum game have played $T$ rounds according to \textbf{Algorithm \ref{alg:HDMDx}} and \textbf{\ref{alg:HDMDy}} with $\eta = \frac{1}{2|A|_{\max}}$. Then
\begin{enumerate}
\item The $\x$-player suffers a $O\l(\frac{\log(T)}{T}\r)$ regret:
\begin{align} 
\max_{\z\in\Delta_m}\sum_{t=3}^T \ip{\z_t - \z}{ -A\w_t} &\leq   { \Big(\log (T-2) + 1 \Big)\Big(20 + \log m +\log n \Big)|A|_{\max}} \\
&= O\l( {\log T}\r) \nn
\end{align}
and similarly for the $\y$-player.
		
\item The strategies $(\z_T,\w_T)$ constitutes an $O\l(\frac{1}{T} \r)$-approximate equilibrium to the value of the game:
\beq 
\left| V -   \ip{\z_T}{ A \w_T} \right| \leq \frac{ \Big(20 + \log m +\log n \Big)|A|_{\max}} {T-2} = O\l(  \frac{1}{T}\r).
\eeq
\end{enumerate}
\end{theorem}
\begin{proof}
Define $\x^*$ as
\begin{align}
\x^* =  \arg\min_{\x \in \Delta_m} \left\langle \x , -A  \l( \frac{1}{T-2}\sum_{t=3}^{T}\y_t \r) \right \rangle.
\end{align}
We define an auxiliary individual regret $\textup{R}_T^\x$ as  
\beq
\textup{R}_T^\x \coloneqq \sum_{t = 3}^{T} \langle \x_t - \x^*, -A\y_t \rangle.
\eeq
Notice that this is the regret on the $\x_t$ sequence versus $\y_t$ sequence, while we are playing $\z_t$'s and $\w_t$'s in the algorithm.

We then have
\begin{align}
\textup{R}_T^\x&= \sum_{t = 3}^{T} \langle \x_t - \x^*, -A\y_t \rangle  \nn\\
&= \langle \x_3 - \x^*,-A \y_3\rangle + \sum_{t=4}^{T} \langle \x_t - \x^* , -A\y_t\rangle \nn\\ 
&\leq 2|A|_{\max} + \sum_{t=4}^{T}\langle \x_t - \x^*, -A\y_t - \g_{t-1}\rangle + \sum_{t=4}^{T}\langle \x_t - \x^*, \g_{t-1}\rangle \nn
\end{align}
where $\g_t \coloneqq -2(t-2)A \w_t+3(t-3)A\w_{t-1} - (t-4)A\w_{t-2}$. Inserting $\w_t = \frac{1}{t-2}\sum_{i = 3}^{t}\y_i$ into the definition of $\g_t$, we get $\g_t = -2A\y_{t} + A\y_{t-1}$. Straightforward calculation then shows:
\begin{align} 
\textup{R}_T^\x&\leq 2|A|_{\max} + \sum_{t=4}^{T}\langle \x_t - \x^*, -A\y_t + 2A\y_{t-1} - A\y_{t-2}\rangle + \sum_{t=4}^{T}\langle \x_t - \x^*, -2A\y_{t-1} + A\y_{t-2}\rangle \nn\\
&=  2|A|_{\max} + \sum_{t=4}^{T}\langle \x_t - \x^*, (-A\y_t + A\y_{t-1}) - (-A\y_{t-1} + A\y_{t-2})\rangle \nn\\
&\qquad \quad \qquad \quad \quad \ \ \ \ \ \ \ \ \ \ \ \ \ \ \ \ \ \ \ \ + \frac{1}{\eta}\sum_{t=4}^{T} \Big(D(\x^*,\x_{t-1}) - D(\x^*,\x_{t}) - D(\x_t,\x_{t-1}) \Big) \nn \\
&= 2|A|_{\max} + \sum_{t=4}^{T-1}\langle \x_t - \x_{t+1}, -A\y_{t} + A\y_{t-1}\rangle + \langle \x_4 - \x^*, A\y_3 - A\y_2\rangle \nn \\
 &\quad  \ \ \ \ + \langle \x_T - \x^* , -A\y_T + A\y_{T-1}\rangle + \frac{1}{\eta}\sum_{t=4}^{T} \Big(D(\x^*,\x_{t-1}) - D(\x^*,\x_{t}) - D(\x_t,\x_{t-1}) \Big) \nn \\
&\leq 10|A|_{\max}+ \sum_{t=4}^{T-1}\langle \x_t - \x_{t+1}, -A\y_{t} + A\y_{t-1}\rangle \nn \\ 
& \hspace{150pt} + \frac{1}{\eta}\sum_{t=4}^{T} \Big(D(\x^*,\x_{t-1}) - D(\x^*,\x_{t}) - D(\x_t,\x_{t-1}) \Big) \nn  \\
&\leq 10|A|_{\max} + \sum_{t=4}^{T-1}\|\x_t - x_{t+1}\|_1 \cdot |A|_{\max} \cdot\| \y_{t} - \y_{t-1}\|_1 \nn \\
&\hspace{150pt} + \frac{1}{\eta}\Big(D(\x^*,\x_3) - D(\x^*,\x_T) \Big) +\sum_{t=4}^{T}\frac{-1}{\eta} D(\x_t, \x_{t-1})  \nn \\
&\leq 10|A|_{\max} + \frac{1}{2}\sum_{t=4}^{T-1} \Big( |A|_{\max} \cdot \|\x_t - \x_{t+1}\|_1^2 + |A|_{\max} \cdot \|\y_{t} - \y_{t-1} \|_1^2 \Big) \nn\\
&\hspace{150pt} + \frac{1}{\eta}\Big(D(\x^*, \x_3) - D(\x^*,\x_T) \Big) +\sum_{t=4}^{T}\frac{-1}{\eta} D(\x_t, \x_{t-1}). \nn 
\end{align} 
Using the fact that $\psi$ is 1-strongly convex with respect to the $\ell_1$-norm, we have $ -D(\x,\x') \leq -\frac{1}{2}\|\x-\x' \|_1^2 \leq 0$. Also, we have $D(\x^*,\x_3) \leq \log m$. Combining these facts in the last inequality gives:
\begin{align*}  
\textup{R}_T^\x&\leq 10|A|_{\max} + \frac{\log m}{\eta} + \frac{|A|_{\max}}{2}\sum_{t=4}^{T-1}   \|\x_t - \x_{t+1}\|_1^2 \nn \\
&\hspace{70pt}  + \frac{|A|_{\max}}{2}\sum_{t=4}^{T-1}   \|\y_t - \y_{t-1}\|_1^2 - \frac{1}{2\eta} \sum_{t=4}^{T} \|\x_{t-1}-\x_t \|_1^2.
\end{align*} 
Similarly, for the second player we define 
\beq
\textup{R}^\y_T \coloneqq \sum_{t=3}^T \ip{\y_t - \y^*}{ A^\top\x_t}
\eeq
where $\y^* \coloneqq \arg\min_\y \ip{\y}{ A^\top \l( \frac{1}{T-2}\sum_{t=3}^T \x_t \r)}$. We then have
\begin{align*} 
\textup{R}^\y_T &\leq 10|A|_{\max} + \frac{\log n}{\eta} + \frac{|A|_{\max}}{2}\sum_{t=4}^{T-1}   \|\y_t - \y_{t+1}\|_1^2 \nn \\
&\hspace{70pt}  + \frac{|A|_{\max}}{2}\sum_{t=4}^{T-1}   \|\x_t - \x_{t-1}\|_1^2 - \frac{1}{2\eta} \sum_{t=4}^{T} \|\y_{t-1}-\y_t \|_1^2.
\end{align*}
Setting $\eta = \frac{1}{2|A|_{\max}}$, we get
\beq \label{eq:HDMD_proof_hold1}
\textup{R}_T^\x+ \textup{R}^\y_T \leq   \Big(20 + \log m +\log n \Big){|A|_{\max}}.
\eeq
Now, recalling that $\z_T = \frac{\sum_{t=3}^T \x_t}{T-2}$ and $\w_T = \frac{\sum_{t=3}^T \y_t}{T-2}$ and using the definition of $\textup{R}_T^\x$ and $\textup{R}^\y_T$, we get
\begin{align} \label{eq:HDMD_proof_hold2}
\frac{1}{T-2}  \Big( \textup{R}_T^\x+ \textup{R}^\y_T  \Big) = \max_{\x \in \Delta_m} \ip{\x}{ A\w_T} - \min_{\y\in\Delta_n} \ip{\z_T}{A\y}.
\end{align}
Furthermore, by the definition of the value of the game,
we have
\beq \label{eq:HDMD_proof_hold3}
\min_{\y \in \Delta_n} \ip{\z_T}{A\y} \leq V \leq \max_{\x \in \Delta_m} \ip{\x}{A\w_T}.
\eeq
We also trivially have 
\beq \label{eq:HDMD_proof_hold4}
\min_{\y \in \Delta_n} \ip{\z_T}{A\y} \leq \ip{\z_T}{A\w_T} \leq \max_{\x \in \Delta_m} \ip{\x}{A\w_T}.
\eeq
Combining \eqref{eq:HDMD_proof_hold2} - \eqref{eq:HDMD_proof_hold4} in \eqref{eq:HDMD_proof_hold1} then establishes \eqref{eq:HDMD_value_of_the_game}:
\beq \nn
\left |V - \ip{\z_T}{A\w_T} \right | \leq  \frac{ \Big(20 + \log m +\log n \Big)|A|_{\max}} {T-2}.
\eeq

We now turn to \eqref{eq:HDMD_low_regret}. 

Let $\textup{R}_T^\z \coloneqq  \max_{\z\in \Delta_m}\sum_{t=3}^T \ip{\z_t - \z}{ -A\w_t}$ and let $\tilde{\textup{R}}_T^\z \coloneqq  \sum_{t=3}^T \ip{\z_t - \z_t^*}{ -A\w_t}$ where $\z^*_t = \arg\min_{\z\in\Delta_m} \ip{\z}{-A\w_t}$. Evidently we have $\textup{R}_T^\z \leq \tilde{\textup{R}}_T^\z$. Notice that (with $\w_t^*$ similarly defined)
\begin{align}
\ip{\z_t - \z_t^*}{ -A\w_t} &= \ip{\z^*_t}{ A\w_t} - \ip{\z_t}{ A\w_t} \nn \\
&\leq \ip{\z^*_t}{ A\w_t} - \ip{\z_t}{ A\w^*_t} \nn\\
&\leq \frac{\Big(20 + \log m +\log n \Big){|A|_{\max}}}{t-2} \label{eq:HDMD_proof_hold5}
\end{align}by \eqref{eq:HDMD_proof_hold1} and \eqref{eq:HDMD_proof_hold2}. Using these inequalities, we get
\begin{align*}
\frac{1}{T-2} \textup{R}_T^\z  \leq \frac{1}{T-2} \tilde{\textup{R}}_T^\z  &= \frac{1}{T-2}\sum_{t=3}^T \ip{\z_t - \z^*_t}{ -A\w_t} \\
&\leq  \frac{1}{T-2} \sum_{t=3}^T \frac{ \Big(20 + \log m +\log n \Big)|A|_{\max}} {t-2} \\
&\leq \frac{ \Big(\log (T-2) + 1 \Big)\Big(20 + \log m +\log n \Big)|A|_{\max}} {T-2} 
\end{align*}which finishes the proof.
\end{proof}
\section{Robust Optimistic Mirror Descent}\label{app:RDMD_proof}
In this appendix, we prove \textbf{Theorem \ref{thm:RDMD}}, repeated below for convenience.

\begin{theorem}[$O(\sqrt{T})$-Adversarial Regret]
Suppose that $\| \nabla f_t \|_* \leq G$ for all $t$. Then playing $T$ rounds of \textbf{Algorithm \ref{alg:RDMD}} with $\eta_t = \frac{1}{G\sqrt{t}}$ against an arbitrary sequence of convex functions has the following guarantee on the regret:
\begin{align}
\max_{\x\in\Delta_m}\sum_{t=1}^T \ip{\x_t - \x}{ \nabla f_t (\x_t)} &\leq   G\sqrt{T}\l( 18+2D^2\r)  + GD\l(3\sqrt{2} + 4D \r) \nn\\
& = O\l( {\sqrt{T} }\r).  \nn
\end{align}
\end{theorem}

\def\myconst{{ 9 }}
\begin{proof}
Define $\textup{R}^\x_T \coloneqq \sum_{t=1}^T \ip{\x_t - \x^*}{ \nabla f_t(\x_t)} $ where $\x^* \coloneqq \arg\min_{\x\in\Delta_m} \ip{\x}{\sum_{t=1}^T \nabla f_t(\x_t)}$. Let $\tilde{\nabla }_t = 2\nabla f_t(\x_t)-\nabla f_{t-1}(\x_{t-1})$, and let $\eta_t = \frac{1}{\alpha \sqrt{t}}$ for some $\alpha >0$ to be chosen later. Then
\begin{align*}
	\textup{R}^\x_T &= \sum_{t=1}^{T} \langle \x_t - \x^*, \nabla f_t(\x_t)\rangle \\
	&\leq \sqrt{2}DG + \sum_{t=2}^{T} \langle \x_t - \x^*, \nabla f_t(\x_t) - \tilde{\nabla}_{t-1}\rangle +\sum_{t=2}^{T} \langle \x_t - \x^*, \tilde{\nabla}_{t-1}\rangle  \\
	&\leq \sqrt{2}DG + \sum_{t=2}^{T} \langle \x_t - \x^*, \nabla f_t (\x_t)- \nabla f_{t-1}(\x_{t-1})\rangle - \sum_{t=2}^{T} \langle \x_t - \x^*, \nabla f_{t-1}(\x_{t-1}) - \nabla f_{t-2}(\x_{t-2})\rangle +\sum_{t=2}^{T} \langle \x_t - \x^*, \tilde{\nabla}_{t-1}\rangle  \\
	&\leq 3\sqrt{2}DG +	\sum_{t=2}^{T-1} \langle \x_t - \x_{t+1},  \nabla f_t(\x_t) - \nabla f_{t-1}(\x_{t-1})\rangle + \sum_{t=2}^{T}\frac{1}{\eta_t}\Big(D(\x^*,\tilde{\x}_{t-1})-D(\x^*, \x_{t}) - D(\x_{t},\tilde{\x}_{t-1}) \Big) \\
	&\leq 3\sqrt{2}DG  + \sum_{t = 2}^{T-1} \l(\frac{\sqrt{t}G}{9}\|\x_t - \x_{t+1}\|^2 + \frac{9 G}{\sqrt{t}} \r)  \\ 
& \hspace{100pt} + \alpha\sum_{t=1}^{T}\sqrt{t}\Big(D(\x^*,\tilde{\x}_{t-1})-D(\x^*, \x_{t}) - D(\x_{t},\tilde{\x}_{t-1}) \Big) .
	\end{align*}
Using the joint convexity of $D(\x,\y)$ in $\x$ and $\y$ and the strong convexity of the entropic mirror map, we get:
\begin{align}
-D(\x_{t}, \tilde{\x}_{t-1}) &\leq - \frac{1}{2} \| \tilde{\x}_t -\x_{t+1} \|^2 \nn\\
&\leq -\frac{1}{4} \left\|\frac{t-1}{t} (\x_t - \x_{t+1}) \right\|^2 +\frac{1}{2} \l(\frac{1}{t}\r)^2 \left\|\x_c - \x_{t+1} \right\|^2 \nn\\
&\leq -\frac{(t-1)^2}{4 t^2}\| \x_t - \x_{t+1}\|^2 + \frac{D^2}{t^2} , \nn
\end{align}
and
\begin{align}
&D(\x^*,\tilde{\x}_t) \leq \frac{t-1}{t}D(\x^*,\x_t)+ \frac{1}{t}D \l(\x^*,\x_c \r).\nn
\end{align}
Meanwhile, straightforward calculations show that
\begin{align}
&\sum_{t=2}^T \frac{D \l(\x^*, \x_c \r)}{\sqrt{t}}  \leq 2D^2\sqrt{T},\nn
\end{align}
and
\begin{align}
\sum_{t=2}^T \l(\sqrt{t} \cdot \frac{t-1}{t} D(\x^*, \x_{t-1}) - \sqrt{t} D(\x^*, \x_t)\r) &\leq \sum_{t=2}^T \l(\sqrt{t-1} D(\x^*, \x_{t-1}) - \sqrt{t} D(\x^*, \x_t)\r) \nn\\
&\leq D(\x^*, \x_1) \leq D^2.\nn
\end{align}
We can hence continue as
	\begin{align}
	\textup{R}^\x_T 	&\leq 
	3\sqrt{2}DG + \sum_{t = 2}^{T-1} \l(\frac{\sqrt{t}}{9}G\|\x_t - \x_{t+1}\|^2 + \frac{ 9G}{\sqrt{t}} \r)  + 2\alpha D^2\sqrt{T}\nn \\ 
	& +\alpha D^2 - \frac{ \alpha }{4}\sum_{t=2}^{T}\sqrt{t} \cdot \l(\frac{t-1}{t}\r)^2 \|\x_{t-1} - \x_{t}\|^2 + \alpha D^2\sum_{t=2}^{T}\frac{\sqrt{t}}{t^2}.  \label{eq:RDMD_proof_hold1} 
\end{align}
Elementary calculations further show
\begin{align}
&\sum_{t = 2}^{T-1}\frac{9G}{\sqrt{t}} \leq 18G\sqrt{T} \nn, \\
& \sum_{t=2}^T \frac{1}{t\sqrt{t}} \leq 3. \nn
\end{align}
Finally, since $(\frac{t-1}{t})^2 \geq \frac{4}{9}$ for $t\geq 3$, we can further bound \eqref{eq:RDMD_proof_hold1} as
\begin{align*}
\textup{R}^\x_t 	&\leq 3\sqrt{2}DG + 18G \sqrt{T} + 2\alpha D^2\sqrt{T} + 4\alpha D^2   \\
& \hspace{80pt} + \l(  \frac{G}{9}\sum_{t=2}^{T-1}  \sqrt{t} \| \x_t - \x_{t+1} \|^2 - \frac{\alpha}{4} \cdot \frac{4}{9}\sum_{t=2}^{T-1} \sqrt{t+1} \| \x_t - \x_{t+1}\|^2 \r).
\end{align*}The proof is finished by choosing $\alpha = G.$
\end{proof}

\bibliographystyle{icml2018}

\end{document}